\documentclass[12pt]{article}
\usepackage[left=2cm, right=2cm, top=2cm]{geometry}
\usepackage{appendix}
\usepackage{endnotes}
\let\footnote=\endnote

\usepackage{caption}
\usepackage{array}
\usepackage{amsthm}
\usepackage{authblk}
\usepackage{graphicx}
\usepackage[lofdepth,lotdepth]{subfig}
\usepackage{dsfont}
\usepackage{amssymb}
\usepackage{amsmath,amssymb,xspace,multirow,multicol,diagbox}
\usepackage{array}
\usepackage{makecell}

\setcellgapes{3pt}
\makegapedcells

\usepackage{algorithm,algpseudocode}

\usepackage{ bbold }

\def\sse{\subseteq}
\def\E{\mathbb{E}}
\def\opt{\ensuremath{\mathsf{OPT}}\xspace}
\def\alg{\ensuremath{\mathsf{ALG}}\xspace}

\usepackage[utf8]{inputenc}
\usepackage[english]{babel}
 
\newtheorem{theorem}{Theorem}[section]
\newtheorem{corollary}{Corollary}[theorem]
\newtheorem{lemma}[theorem]{Lemma}

\newcommand{\ignore}[1]{}

\begin{document}

\title{Approximation Algorithms for the A Priori Traveling Repairman}

\author[1]{Inge Li G{\o}rtz\thanks{inge@dtu.dk}}
\author[2]{Viswanath Nagarajan\thanks{viswa@umich.edu}}
\author[2]{Fatemeh Navidi\thanks{navidi@umich.edu}}
\affil[1]{Technical University of Denmark, DTU Compute, Denmark}
\affil[2]{University of Michigan,  Ann Arbor, MI, USA}

\renewcommand\Authands{ and }
\date{}
\maketitle
\begin{abstract}
We consider the {\em a priori} traveling repairman problem, which is a stochastic version of the classic traveling repairman   problem (also called the traveling deliveryman or minimum latency problem). Given a metric $(V,d)$ with a root $r\in V$, the traveling repairman problem (TRP) involves finding a tour originating from $r$  that minimizes the sum of arrival-times at all vertices. In its {\em a priori} version, we are also given   independent probabilities of  each  vertex being active. We want to find a master  tour $\tau$ originating from $r$ and visiting all vertices. The objective is to  minimize  the  expected sum of arrival-times at all active vertices, when $\tau$ is shortcut over the inactive vertices. 
We obtain the first constant-factor approximation algorithm for {\em a priori} TRP  under non-uniform probabilities. Previously, such a result was only known for uniform probabilities. 
\end{abstract}

\begin{keywords}{Traveling Repairman Problem, A Priori Optimization, Approximation Algorithms}\end{keywords}

\section{Introduction}

Traditional optimization models assume full information on the instances being solved, which is unrealistic in many situations. In order to remedy this limitation,   there has been significant work in the area of optimization under uncertainty, which deals with various ways to model uncertain input. Stochastic optimization is a widely used approach, where one models the input probabilistically.

{\em A priori} optimization~(\cite{bertsimas1990priori}) is an elegant model for stochastic combinatorial optimization, that is particularly useful when one needs to repeatedly solve instances of the same optimization problem. The basic idea here  is to  reduce the computational overhead of solving repeated problem instances by performing suitable pre-processing  using distributional information. More specifically, in an {\em a priori} optimization problem, one is given a  probability distribution $\Pi$  over inputs and the goal is to find a ``master solution'' $\tau$. Then, after observing the random input $A$ (drawn from the distribution $\Pi$), the master solution $\tau$ is modified using a simple rule to obtain a solution $\tau_A$ for that particular input. The objective is to minimize the expected value of the master solution. For a problem with objective function $\phi$, we are interested in:  
$$\min_{\tau:\mbox{master solution}} \quad \E_{A\gets \Pi}\,[\phi(\tau_A)]\; .$$

This paper studies the {\em a priori} traveling repairman problem. The traveling repairman problem (TRP) is a fundamental vehicle routing problem that involves computing a  tour originating from a depot/root  that minimizes the sum of latencies (i.e. the distance  from the root on this tour) at all vertices. The TRP is also known as the traveling deliveryman or minimum latency problem, and has been studied extensively, see e.g. ~\cite{picard1978time}, \cite{FischettiLM93}, \cite{goemans1998improved}. In the {\em a~priori}~TRP, the master solution $\tau$ is a tour visiting all vertices, and for any  random input (i.e. subset $A$ of vertices), the solution $\tau_A$ is simply obtained by visiting the vertices of $A$ in the order given by $\tau$. 

An {\em a priori} solution is advantageous in settings when we  repeatedly solve  instances of the TRP that are drawn from a common distribution. For example, we may need to solve one TRP instance on each day of operations, where the distribution over  instances  is estimated from historical data.  Using  an {\em a priori} solution saves on computation time as we do not have to solve each instance from scratch. Moreover, for vehicle routing problems (VRPs) there are also practical advantages to have a pre-planned master tour, e.g. drivers have familiarity with the  route followed each day. See \cite{SG95}, \cite{CT08}, and \cite{ESU09} for more discussion on the  benefits of a pre-planned VRP solution.

\subsection{Problem Definition.}
\def\lat{\mathsf{LAT}}
\def\elat{\mathsf{ELAT}}

The traveling repairman problem (TRP) is defined on a finite metric 
$(V,d)$ where $V$ is a vertex set and $d: V \times V\rightarrow \mathbb{R}_+$ is a distance function. We assume that the distances are symmetric and satisfy triangle inequality. There is also a designated root vertex $r\in V$. The goal is to find a tour $\tau$ originating from $r$ that  visits all vertices. The {\em  latency} of any vertex $v$ in tour $\tau$ is the length of the path from $r$ to $v$ along  $\tau$. The objective in TRP is to minimize the sum  of latencies of all vertices.

In the {\em a priori} TRP, in addition to the above input we are also given activation probabilities $\{p_v\}_{v\in V}$ at all vertices; we use $\Pi$ to denote this distribution. In this paper, as in most prior works on {\em a priori} optimization, we assume that the input distribution $\Pi$ is independent accross vertices. So the active subset $A\sse V$ contains each vertex $v\in V$ independently with probability $p_v$. A solution to {\em a priori} TRP is a master tour $\tau$ originating from $r$ and visiting all vertices. Given an active subset $A\sse V$, we restrict tour $\tau$ to vertices in $A$ (by shortcutting over $V\setminus A$) to obtain tour $\tau_A$, again originating from $r$.  For each $v\in A$, we use $\lat_{\tau}^A(v)$ to denote the latency of $v$ in tour $\tau_A$. We also use $\lat_{\tau}^A = \sum_{v\in A} \lat_{\tau}^A(v)$ for the total latency under active subset $A\sse V$. The objective is to minimize 
$$\elat_{\tau} = \E_{A\gets \Pi}\left[ \lat_{\tau}^A\right] = \E_{A\gets \Pi}\left[  \sum_{v\in A} \lat_{\tau}^A(v)\right]\;.$$

\subsection{Results.}
Our main result in this note is the first constant-factor approximation for the {\em a priori} TRP. 
\begin{theorem}\label{thm:apTRP-gen}
There is a constant-factor approximation algorithm for the {\em a priori} traveling repairman problem under independent probabilities. 
\end{theorem}
Previously, \cite{van2018priori} obtained such a result  under the restriction that all activation probabilities are {\em identical}, and posed the general case of non-uniform probabilities as an open question-- which we resolve. Our result adds to the  small list of {\em a priori} VRPs with provable worst-case guarantees: traveling salesman, capacitated vehicle routing and traveling repairman. 

In fact, we obtain Theorem~\ref{thm:apTRP-gen} by a generic reduction of {\em a priori} TRP from non-uniform to uniform probabilities, formalized below. 

\begin{theorem}\label{thm:apTRP-redn}
There is a $(6.27\rho)$-factor approximation algorithm for the {\em a priori} traveling repairman problem under independent probabilities, where $\rho$ denotes the best approximation ratio for the problem under uniform probabilities.
\end{theorem}

Clearly, Theorem~\ref{thm:apTRP-gen} follows by combining Theorem~\ref{thm:apTRP-redn} with the   $O(1)$-approximation algorithm for {\em a priori} TRP under uniform probabilities by~\cite{van2018priori}. As the constant factor in \cite{van2018priori} for uniform probabilities is quite large, there is the possibility of improving it using a different algorithm: Theorem~\ref{thm:apTRP-redn} would be applicable to any such future improvement and yield a corresponding improved result for non-uniform probabilities.  

\subsection{Related Work.}

The {\em a priori} optimization model was introduced in \cite{Jaillet85} and \cite{bertsimas1988probabilistic}, see also the survey by~\cite{bertsimas1990priori}. These papers considered  the setting where  the metric is itself random and carried out asymptotic analysis (as the number of vertices grows large). They obtained such results for the minimum spanning tree, traveling salesman, capacitated vehicle routing and traveling salesman facility location problems.

Approximation algorithms for {\em a priori} optimization are more recent: these can handle arbitrary problem instances. Such results are known for the traveling salesman problem (TSP), capacitated VRP and traveling repairman (TRP). We briefly discuss them below. 

The {\em a priori} TSP has been extensively studied. In particular, there is a randomized $4$-approximation algorithm for independent probabilities by~\cite{shmoys2008constant}. The same paper also gave a deterministic $8$-approximation algorithm; the constant was later improved to $6.5$ in \cite{van2011deterministic}. These algorithms were based on a random-sampling approach~(\cite{GuptaKPR07,williamson2007simpler}) that was previously used in other network design problems. For arbitrary (black-box) distributions, 
\cite{schalekamp2008algorithms} gave a randomized $O(\log n)$-approximation algorithm which actually does not even need any knowledge on the distribution. Later, \cite{gorodezky2010improved} proved an $\Omega(\log n)$ lower bound on the approximation ratio of any deterministic algorithm for {\em a priori} TSP under arbitrary distributions.  

The capacitated VRP with stochastic demands~(\cite{b92}) is another well-studied {\em a priori} optimization problem. Here, we have a vehicle with limited capacity $Q$ at the root that needs to satisfy demands at various vertices. The demand at each vertex is an independent random variable with a known distribution. A master solution to this problem is a tour $\tau$ that visits every vertex;  after  demands are observed, the vehicle visits  vertices in the same order as $\tau$ while  performing additional refill-trips to the root whenever it runs out of items. The objective is to minimize the total length of the tour. A $2.5$-approximation algorithm for this problem in the case of {\em identical} demand distributions was given in \cite{b92}. Later,  \cite{GNR12}  obtained a randomized $2.5$-approximation algorithm for this problem under non-identical distributions. 

The {\em a priori} TRP was recently studied in \cite{van2018priori}, where a constant-factor approximation algorithm was obtained for the case of uniform independent probabilities. They left open the problem under non-uniform probabilities: Theorem~\ref{thm:apTRP-redn} resolves this positively.  The algorithm in \cite{van2018priori} was based on many ideas from the deterministic TRP, but it needed stochastic counterparts of various properties. As noted in \cite{van2018priori}, their proof  relied heavily on the  probabilities being uniform and it was unclear how to handle non-uniform  probabilities. 

We note that the deterministic traveling repairman problem (TRP) has  been studied extensively, both in exact algorithms~(\cite{picard1978time,FischettiLM93,wu2000polynomial}) and approximation algorithms (\cite{blum1994minimum}, \cite{goemans1998improved}, \cite{arora2003approximation}, \cite{chaudhuri2003paths}). It was shown to be NP-hard even on weighted trees by~\cite{sitters2002minimum}, and a polynomial time approximation scheme (PTAS) on such metrics was given by~\cite{Sitters14}. On general metrics, the best approximation ratio known is $3.59$ due to~\cite{chaudhuri2003paths}; it is also known that one cannot obtain a PTAS. 

\section{{\em A Priori} TRP with Non-Uniform Distribution}

\def\ii{\ensuremath{{\cal I}}\xspace}
\def\ij{\ensuremath{{\cal J}}\xspace}
\def\ik{\ensuremath{{\cal I}_q}\xspace}
\def\ib{\ensuremath{{\cal I}_{\bar{p}}}\xspace}
\def\hV{\ensuremath{\widehat{V}}\xspace}   

Consider an instance $\mathcal{I}$ of {\em a priori} TRP on metric $(V,d)$ with probabilities $\{p_v\}_{v\in V}$. We show how to  ``reduce''  this  instance to one with uniform probabilities, which would prove Theorem~\ref{thm:apTRP-redn}.  Our approach is natural: we  replace  each vertex $v\in V$ with a group $S_v$ of co-located vertices, where each new vertex is active with a uniform probability $p$. Let \ij denote the new instance and $(\hV,d)$ the new metric. Intuitively, when  $p$ is chosen much smaller than the $p_v$s and $|S_v| \approx p_v/p$, the scaled uniform instance \ij should behave similar to ${\cal I}$. However, proving such a result formally requires significant technical work. For example,  
the master tour found by an algorithm for the  scaled (uniform) instance  might not visit all the co-located copies consecutively. We define a {\em consecutive} master tour for \ij as one that visits all co-located vertices consecutively. Then, we show an approximate equivalence between (i) master tours in \ii and (ii) consecutive master tours in \ij. 
This relies on the independence across vertices and the  correspondence between the events ``vertex $v$ is active in \ii'' and ``at least one vertex of $S_v$ is active in \ij''. This is formalized in Section~\ref{subsec:x}. Then, we show in Section~\ref{subsec:consec} that any master tour for instance \ij  can be modified to a ``consecutive'' master tour with the same or better overall expected latency. Finally, in order  to maintain a polynomial-size instance \ij  (this is reflected in the choice of $p$),  we need to take care of  vertices with very small probability separately. In Section~\ref{subsec:y} we show that the overall effect of the  small-probability vertices is tiny if they are  visited in non-decreasing order of distances   at the end  of our master tour. 

\begin{algorithm}
\caption{Reducing non-uniform instance \ii to uniform instance \ij}\label{const}
\begin{algorithmic}[1]
\State $Y \gets \{v\in V\,|\,p_v<1/{n^2}\}$ denotes the low probability vertices.
\State $X \gets \{v\in V \,|\, p_v\geq 1/{n^2}\}$ denotes all other vertices. 
\State $p \gets \frac{1}{n}\min_{v\in X}p_v$
\State Construct instance $\mathcal{J}$ with vertex set \hV  that contains for each $v\in X$, a set $S_v$ of $t_v = \lceil\frac{p_v}{p}\rceil$ copies of $v$. The distance between any two vertices of $S_v$ is zero for all $v\in X$. The distance between any vertex of $S_u$ and any vertex of $S_v$ is $d(u,v)$. All vertices in \hV have a uniform activation probability $p$.
\State Run any approximation algorithm for {\em uniform a priori} TRP  on $\mathcal{J}$ to obtain master tour $\widehat{\pi}$.
\State\label{step:consec}Run procedure \textsc{MakeConsecutive}($\widehat{\pi}$) to ensure that $\widehat{\pi}$ visits each group $S_v$ consecutively.
\State Obtain tour $\pi$ by visiting vertices of $X$ 
  in the same order that  $S_v$s are visited in ${\widehat{\pi}}$.
\State Extend $\pi$ by visiting vertices $w\in Y$ in {\em non-decreasing} order of $d(r,w)$, to obtain tour $\bar{\pi}$.\\
\Return $\bar{\pi}$.
\end{algorithmic}
\end{algorithm}

Algorithm~\ref{const} describes the reduction formally. In Step~\ref{step:consec}, Algorithm~\ref{const} relies on a  procedure \textsc{MakeConsecutive} that  modifies tour $\widehat{\pi}$  such that it visits all copies of the same node consecutively. We will prove Theorem~\ref{thm:apTRP-redn} by analyzing this algorithm.

\subsection{Overview of Analysis.} 
We first assume that the master tour $\widehat{\pi}$ on instance \ij already  visits copies of each vertex consecutively: so there is no need for Step~\ref{step:consec}. We split this proof into two parts corresponding to the $X$-vertices (normal probabilities) and $Y$-vertices (low probabilities). The analysis for $X$-vertices (Section~\ref{subsec:x}) is the main part, where we  show that the optimal values of \ii and \ij are within a constant factor of each other. In Lemma~\ref{lem:probcost-inc}  we show that a constant-factor perturbation in probabilities of $V$ will only change the cost of any solution (including the optimal) by a constant factor. Then we prove (in Lemma~\ref{lem:optineq}) that the optimal value of instance \ij is within a constant factor of the optimal value of  \ii: although \ij has many more vertices than \ii, the proof exploits the fact that the expected number of active vertices is roughly the same as \ii. Lemma~\ref{lem:algineqcomp} proves the other direction for the cost of our algorithm, i.e. the cost of Algorithm~\ref{const} for \ii is at most that of the consecutive master tour for \ij.  To handle the $Y$-vertices,  we use a simple expected distance lower-bound to  show (in Section~\ref{subsec:y})  that visiting $Y$ at the end of  our tour only adds a small factor to the overall expected cost.

Note that we assumed above that the master tour $\widehat{\pi}$ visits copies of each vertex consecutively. It is possible that the algorithm for uniform {\em a priori} TRP  in~\cite{van2018priori} already has this property, in which case the analysis outlined above suffices. However, by  providing an explicit   subroutine (\textsc{MakeConsecutive}) that ensures this consecutive property, our approach can be combined with {\em any}  algorithm for uniform {\em a priori} TRP. The details of the \textsc{MakeConsecutive} procedure and its analysis appear in Section~\ref{subsec:consec}.  

\subsection{Analysis for vertices in X.}\label{subsec:x}

Here we analyze the steps of the algorithm that deal with vertices in $X$, i.e. with probability at least $\frac1{n^2}$. In order to reduce notation, we will assume here that $X=V$ which is the entire vertex set.  Recall that $p=\frac1n\cdot \min_{v\in V} p_v$ . 
 Also define $\bar{p}_v=\min\left\{ (1+\frac1n)p_v, 1\right\}$, $t_v=\lceil p_v/p\rceil$ and $q_v=1-(1-p)^{t_v}$ for each $v\in V$. We will refer to the instances on metric $(V,d)$ with probabilities $\{p_v\}_{v\in V}$, $\{q_v\}_{v\in V}$ and $\{\bar{p}_v\}_{v\in V}$ as $\ii_p$, \ik and \ib respectively. Note that the original instance is $\ii=\ii_p$. For simplicity we use $\mathbf{p}, \mathbf{q}$ and $\mathbf{\bar{p}}$ to refer to the vector of probabilities for each corresponding distribution.

\begin{lemma}For any $v\in V$, we have $p_v(1-\frac{1}{e})\leq q_v\leq \bar{p}_v\le  p_v (1+\frac{1}{n})$.\label{lem:alphatimes}\end{lemma}
\begin{proof}
Note that for every real number $x$ we have $1+x\leq e^x$: using $x = -p$ and raising both sides to the power of $t_v$ we obtain $(1-p)^{t_v}\leq e^{-pt_v}$. Now we have:
$$q_v = 1 - (1-p)^{t_v} \geq 1 - e^{-pt_v}\geq 1-e^{-p\cdot\frac{p_v}{p}} = 1 - e^{-p_v} \geq (1-\frac{1}{e})p_v\;.$$
The second inequality uses $t_v = \lceil p_v/p\rceil$ and  the last one uses $1-e^{-x} \ge (1-1/e)x$ for any $x\in [0,1]$ with $x=p_v$.
Now, to prove the other inequality we consider the bionomial expansion of $(1-p)^{t_v}$ and cut it off for the powers greater than 1. So we  have:
$$q_v = 1- (1-p)^{t_v}\leq 1 - (1-pt_v)=pt_v\leq p(\frac{p_v}{p}+1)\leq p_v + \frac{p_v}{n} = p_v(1+\frac{1}{n})\;.$$
Combined with the fact that $q_v\le 1$, we obtain $q_v\le \bar{p}_v$.  
\end{proof}
\begin{lemma}
Let $\pi$ be any master tour on $(V,d)$. Consider two probability distributions given by $\{q_v\}_{v\in V}$ and $\{\bar{p}_v\}_{v\in V}$ such that $0\le q_v \le \bar{p}_v\le 1$ for each $v\in V$. Then the expected latency of $\pi$ under $\{q_v\}_{v\in V}$ is at most that under $\{\bar{p}_v\}_{v\in V}$.  \label{lem:probcost-dec}
\end{lemma}
\begin{proof}
Let function $f(p_1,\cdots p_n)$ denote the expected latency of $\pi$ as a function of vertex  probabilities $\{p_v\}$. We will show that all partial derivatives of $f$ are non-negative. This would imply the lemma. 

We can express $f$ as a multilinear polynomial 
$$f(\mathbf{p}) = \sum_{A\sse V}  \left( \prod_{u\in A} p_u  \prod_{w\in V\setminus A} (1-p_w ) \right) \cdot \lat_\pi^A\;.$$ 
Recall that $\lat_\pi^A$ is the total latency of vertices in active set $A$ in the shortcut tour $\pi_A$. 
So the $v^{th}$ partial derivative is:
$$\frac{\partial f}{\partial p_v} = \sum_{A\sse V\setminus v} \left( \prod_{u\in A} p_u  \prod_{w\in V\setminus A\setminus v} (1-p_w ) \right) \big(  \lat_\pi^{A\cup v} - \lat_\pi^{A} \big)\;.$$
For any $A\sse V\setminus v$, it follows by triangle inequality  that $\lat_\pi^{A\cup v} \ge \lat_\pi^A$. This shows that each term in the above summation is non-negative and so $\frac{\partial f}{\partial p_v} \ge 0$.  
\end{proof}

\begin{lemma}
Let $\pi$ be any master tour on $(V,d)$. Consider two probability distributions  given by $\{q_v\}_{v\in V}$ and $\{\bar{p}_v\}_{v\in V}$ and some constant $\beta \leq 1$ such that $\beta \bar{p}_v\le q_v\le\bar{p}_v$ for each $v\in V$. Then the expected latency of $\pi$ under $\{q_v\}_{v\in V}$ is at least $\beta^3$ times that under  $\{ \bar{p}_v\}_{v\in V}$.  
\label{lem:probcost-inc}
\end{lemma}
\begin{proof}
Let function $f(p_1,\cdots p_n)$ denote the expected latency of $\pi$ under probabilities $\{p_v\}_{v\in V}$. For $\mathbf{q}$ and $\mathbf{\bar{p}}$ as in the lemma, we will show $f(\mathbf{q}) \ge \beta^3\cdot f(\mathbf{\bar{p}})$. To this end, we now view $f$ as the expected sum of  terms corresponding to  all possible edges used in the shortcut tour $\pi_A$ (where $A$ is the active set). Renumber the vertices as $1,2,\cdots n$ in the order of appearance in $\pi$; so the root $r$ is numbered $1$. For any $i,j\in [n]$ let $I_{ij}$ denote the indicator random variable for (ordered) edge $(i,j)$ being used in the shortcut tour $\pi_A$. For any $j\in [n]$, let $N_j$ denote the number of active vertices among $\{j,j+1,\cdots n\}$. Then, the total latency of tour $\pi_A$ is 
$$\sum_{1\le i<j\le n} d(i,j)\cdot I_{ij}\cdot N_j\;.$$
Under probabilities $\mathbf{q}$, for any $i<j$ we have  $\E[I_{ij}] = q_i\cdot q_j \cdot \prod_{k=i+1}^{j-1} (1-q_k)$ which corresponds to the event that $i$ and $j$ are active but all vertices between $i$ and $j$ are inactive. Moreover, $\E[N_j | I_{ij}=1] = 1+ \sum_{\ell=j+1}^n q_\ell$ using the independence across vertices. So we can write:
$$f(\mathbf{q}) = \sum_{1\le i<j\le n} d(i,j)\cdot \E[I_{ij}] \cdot \E[N_j|I_{ij}=1] = \sum_{1\le i<j\le n} d(i,j)\cdot q_i\cdot q_j \cdot \prod_{k=i+1}^{j-1} (1-q_k) \left(1+\sum_{\ell=j+1}^n q_\ell\right)\;.$$
Note that for any $i<j$, using the fact that $\beta \cdot \mathbf{\bar{p}} \le \mathbf{q} \le \mathbf{\bar{p}}$  we have:
$$q_i\cdot q_j \cdot \prod_{k=i+1}^{j-1} (1-q_k) \left(1+\sum_{\ell=j+1}^n q_\ell\right) \,\,\ge\,\, \beta^3\cdot  \bar{p}_i\cdot \bar{p}_j \cdot \prod_{k=i+1}^{j-1} (1-\bar{p}_k) \left(1+\sum_{\ell=j+1}^n \bar{p}_\ell\right).$$
This implies $f(\mathbf{q}) \ge \beta^3\cdot f(\mathbf{\bar{p}})$ as desired.  
\end{proof}

\begin{lemma}\label{lem:optineq} Instances $\mathcal{I}$ and  ${\mathcal{J}}$  in Algorithm~\ref{const} satisfy
$$\opt(\mathcal{J})\leq \left(\frac{e}{e-1}\right) \left(1+\frac1n\right)^4\cdot \opt(\mathcal{I})\;.$$
\end{lemma}
\begin{proof}
Recall the three instances $\ii=\ii_p$, $\ik$ and $\ib$ on the metric $(V,d)$. Using $\mathbf{q}\le \mathbf{\bar{p}}$ (Lemma~\ref{lem:alphatimes}) and Lemma~\ref{lem:probcost-dec} we have $\opt(\ik)\le \opt(\ib)$. Further, using 
$\mathbf{p}\le \mathbf{\bar{p}}  \le (1+1/n)\mathbf{p}$ and Lemma~\ref{lem:probcost-inc} we have $\opt(\ib)\le (1+1/n)^3 \opt(\ii_p)$. So we obtain $\opt(\ik)\le (1+1/n)^3 \cdot \opt(\ii)$.

\def\e{{\cal E}}

For $\alpha=\frac{e}{e-1}(1+\frac1{n})$, we will show that $\opt(\ij)\le \alpha\cdot \opt(\ik)$ which would prove the lemma. Recall that instance \ij is defined on the ``scaled'' vertex set $\widehat{V}= \cup_{v\in V} S_v$. Let $\pi$ be an optimal master tour for instance \ik and $\widehat{\pi}$ be its corresponding master tour for \ij: i.e. $\widehat{\pi}$ visits each group $S_v$ consecutively at the point when $\pi$ visits $v$. It suffices to  show that the expected latency $\elat_{\widehat{\pi}}$ of tour  $\widehat{\pi}$  for \ij is at most $\alpha\cdot \elat_\pi$, where $\elat_\pi$ is the expected latency of tour $\pi$ for \ik.  

Let $A\sse V$ and $\widehat{A}\sse \hV$ denote the random active subsets in the instances \ik and \ij respectively. For any $v\in V$,  
let $\e_v$ denote the event that $S_v\cap \widehat{A}\ne \emptyset$; note that these events are independent. Moreover,  for any $v\in V$,    $\Pr_{\widehat{A}}[\e_v] = \Pr_{\widehat{A}}[S_v\cap \widehat{A}\ne \emptyset] = q_v = \Pr_A[v\in A]$. Let $\elat_{\widehat{\pi}}(w) = \E_{A\gets \Pi}\left[  \sum_{v\in S_w} \lat_{\widehat{\pi}}^A(v)\right]$ denote the total expected latency of vertices of $S_w$ in tour $\widehat{\pi}$. Fix any vertex $w\in V$: we will show that   $\elat_{\widehat{\pi}}(w)$   is at most  $\alpha\cdot \elat_\pi(w)$, where $\elat_\pi(w)$ is the expected latency of vertex $w$ in $\pi$.  Summing over $w\in V$, this would imply $\elat_{\widehat{\pi}}\le \alpha\cdot \elat_\pi$, and hence $\opt(\ij)\le \alpha\cdot \opt(\ik)$.  

Consider now a fixed $w\in V$. Note that the probability distribution of the vertices in $V\setminus\{w\}$ whose groups (in $\hV$) have at least one  vertex in $\widehat{A}$ is {\em identical} to that of $A\setminus\{w\}$. In other words, the random subset $\{v\in V\setminus \{w\}: \e_v \mbox{ occurs for } \widehat{A}\subseteq \widehat{V}\setminus S_w\}$ has the same  distribution as random subset  $A\setminus \{w\}$. Below, we couple these two distributions: We {\em condition} on the events $\e_v$  for all $v\in V\setminus\{w\}$ (for tour $\widehat{\pi}$) which corresponds to  conditioning on $A\setminus\{w\}$ being active (for tour $\pi$).  Under this conditioning (denoted $\e$), the latency of any active $S_w$ vertex in $\widehat{\pi}$ is deterministic and equal to the latency of  $w$ (if it is active) in $\pi$; let $L(\pi, w\,|\, \e)$ denote this deterministic value.
So the conditional expected latency of $w$ is $L(\pi, w\,|\, \e)\cdot \Pr[w\in A] = L(\pi, w| \e)\cdot q_w$ where we used the independence of $A\setminus \{w\}$ and the event $w\in A$.  Similarly, the total conditional expected latency of $S_w$ in $\widehat{\pi}$ is $$L(\pi, w| \e )\cdot \E[|\widehat{A}\cap S_w|]  = L(\pi, w| \e)\cdot (p t_w) \le L(\pi, w| \e )\cdot (p_w+p)\;.$$
The equality above  uses the independence of $\{\e_v :  v\in V\setminus \{w\}\}$ and $\widehat{A}\cap S_w$, and the inequality uses $t_w=\lceil p_w/p\rceil$. Thus, the total conditional expected latency of $S_w$ in $\widehat{\pi}$ is at most $\frac{p_w+p}{q_w}$ times the conditional expected latency of $w$ in $\pi$. Deconditioning, we obtain 
$\elat_{\widehat{\pi}}(w) \le \frac{p_w+p}{q_w} \cdot \elat_{ \pi }(w)$. Using Lemma~\ref{lem:alphatimes},   
$\frac{p_w+p}{q_w} \le \frac{e}{e-1}\left(1+p/p_w\right)\le \frac{e}{e-1}(1+1/n)=\alpha$. So $\lat_{\widehat{\pi}}(w) \le \alpha\cdot \lat_{ \pi }(w)$ as needed.~ 

\end{proof}

\begin{lemma}\label{lem:algineqcomp} Consider any consecutive master tour $\widehat{\pi}$ on instance ${\cal J}$ with expected latency $\alg({\cal J})$.  
Then the expected latency of the resulting master tour $\pi$ on instance ${\cal I}$ is 
$$\alg({\cal I}) \le \left(\frac{e}{e-1}\right)^3\left(1+\frac1{n}\right)^3 \cdot \alg({\cal J})\;.$$   
\end{lemma}
\begin{proof}
Let $\alg(\ii_p)$, $\alg(\ik)$ and $\alg(\ib)$ denote the expected latency of master tour $\pi$ under probabilities $\mathbf{p}$, $\mathbf{q}$ and $\mathbf{\bar{p}}$ respectively.  Below we use $\alpha=\frac{e}{e-1}(1+\frac1{n})$.   Using  $\mathbf{p}\le \mathbf{\bar{p}}$  and Lemma~\ref{lem:probcost-dec} we have $\alg(\ii_p)\le \alg(\ii_{\bar{p}})$. Using $\frac1\alpha\cdot \mathbf{\bar{p}} \le \mathbf{q}\le \mathbf{\bar{p}}$ (Lemma~\ref{lem:alphatimes}) and Lemma~\ref{lem:probcost-inc}, we have $\alg(\ii_{\bar{p}}) \le  \alpha^3 \cdot \alg(\ik)$. Combining these bounds, we have $\alg(\ii)\le  {\alpha^3}\cdot \alg(\ik)$. Finally, it is easy to see that $\alg(\ik)\le\alg(\ij)$ as the probability of having at least one active vertex in  group $S_v$ (for any $v\in V$) in \ij is exactly equal the probability ($q_v$) of visiting $v$ in \ik.  
\end{proof}

\subsection{Overall analysis including  vertices in $Y$.}  \label{subsec:y}

Now we have the  tools to finish the proof of Theorem~\ref{thm:apTRP-redn} assuming the tour $\widehat{\pi}$ in \ij is consecutive. Recall that $\pi$ is the tour corresponding to $\widehat{\pi}$ on vertices $X$ and $\bar{\pi}$ is the extended tour that also visits the vertices $Y$. 

First, the analysis for the vertices $X$ (Lemmas \ref{lem:optineq} and \ref{lem:algineqcomp}) yields:
\begin{corollary}\label{cor:normal-lat}
The tour $\pi$ on vertices $X$ satisfies 
$$\E_A\left[\sum_{v\in A\cap X} \lat^A_\pi(v)\right] \le (1+o(1))\left(\frac{e}{e-1}\right)^4 \rho \cdot \opt_X\;,$$ 
where $\rho$ is the approximation ratio for uniform {\em a priori} TRP and $\opt_X$ is the optimal value of the instance restricted to vertices $X$. 
\end{corollary}

After extending tour $\pi$ to $\bar{\pi}$, we can write the final expected latency as  
{\small\begin{equation} \label{eq:overall-lat}
\alg({\cal I})  = \E_{A}\left[\sum_{v\in A\cap X}\lat_{\bar{\pi}}^{A}(v) + \sum_{v\in A\cap Y}\lat_{\bar{\pi}}^{A}(v)\right] =  \E_{A}\left[\sum_{v\in A\cap X}\lat_{\pi}^{A}(v)\right] + \E_{A}\left[\sum_{v\in A\cap Y}\lat_{\bar{\pi}}^{A}(v)\right]
\end{equation}}
where $A\sse V$ is the active subset. The last equality uses the fact that $\bar{\pi}$ visits all vertices of $X$ (along $\pi$) before $Y$. The first term above can be bounded 
by Corollary~\ref{cor:normal-lat}. We now focus on the second term involving vertices $Y$. 

Let $L$ denote the length of tour $\bar{\pi}$ before visiting the first $Y$-vertex; note that this is a random variable. Clearly  $\E[L]$ is at most the expected total latency of the $X$-vertices. Consider any $v\in Y$:  by the ordering of the $Y$-vertices in   master tour   $\bar{\pi}$,
$$\lat^A_{\bar{\pi}}(v) \le \left(L + (2N_v+1)\cdot d(r,v)\right)\cdot \mathbf{1}_{v\in A}\;,$$
where $N_v$ is the number of active $Y$-vertices   appearing before $v$. Taking expectations, 
\begin{align*}
\E[\lat^A_{\bar{\pi}}(v)] & \le p_v\cdot \E[L] + p_v\cdot d(r,v)\cdot (2 \E[N_v] + 1) \le p_v\cdot \E[L] + p_v\cdot d(r,v)\cdot (2n\cdot \frac1{n^2} + 1) \\
& =  p_v\cdot \E[L] + p_v\cdot d(r,v)\cdot (1+o(1)),
\end{align*}
The first inequality uses the fact that  $L$, $N_v$ and $\mathbf{1}_{v\in A}$ are independent. The second inequality uses that   $N_v$ is the sum of at most $n$ Bernoulli random variables each with probability at most $\frac1{n^2}$.

Summing over all $v\in Y$, we obtain
$$\E_{A}\left[\sum_{v\in A\cap Y}\lat_{\bar{\pi}}^{A}(v)\right]\le \left(\sum_{v\in Y} p_v\right) \cdot \E[L] + (1+o(1))\sum_{v\in Y} p_v\cdot d(r,v)\le \frac1n\cdot \E[L] + (1+o(1))\sum_{v\in Y} p_v\cdot d(r,v)\;,$$
where the last inequality uses $p_v\le 1/n^2$ for all $v\in Y$. 

Let $E_X$ denote the expected latency of the $X$-vertices: this is the first term in the right-hand-side of \eqref{eq:overall-lat}. Recall that $\E[L]\le E_X$. Using the above bound on the latency of $Y$-vertices, 
\begin{eqnarray}
\alg({\cal I}) &\le& E_X + \frac1n \cdot E_X + (1+o(1))\sum_{v\in Y} p_v\cdot d(r,v) = (1+o(1))\left( E_X + \sum_{v\in Y} p_v\cdot d(r,v) \right) \notag \\
&\le & (1+o(1))\left(\frac{e}{e-1}\right)^4 \rho \cdot \left( \opt_X + \sum_{v\in Y} p_v\cdot d(r,v) \right)\label{eq:alg-bound1} \\
&\le & (1+o(1))\left(\frac{e}{e-1}\right)^4 \rho \cdot \opt.\label{eq:alg-bound2}
\end{eqnarray}
Above, inequality~\eqref{eq:alg-bound1} uses Corollary~\ref{cor:normal-lat}. Inequality \eqref{eq:alg-bound2} uses the fact that the latency contribution of $Y$-vertices in {\em any} master tour is at least $\sum_{v\in Y} p_v\cdot d(r,v)$ and the latency of $X$-vertices is clearly at least $\opt_X$. This completes the proof of Theorem~\ref{thm:apTRP-redn} assuming that $\widehat{\pi}$ visits each group $S_v$ consecutively. The next subsection shows that this consecutive property can always be ensured.

\subsection{Ensuring the Consecutive Property.} \label{subsec:consec}
The main result here is:
\begin{theorem}\label{thm:consec}
Consider any instance ${\cal J}$ of uniform {\em a priori} TRP on vertices $\cup_{v\in X} S_v$ where  the vertices in $S_v$ are co-located for all $v\in X$. There is a polynomial time algorithm that given any master tour $\tau$, modifies it into  a consecutive  tour having expected latency at most that of $\tau$.
\end{theorem}
While this result is intuitive, we note that it is not obvious to prove. This is because an optimal TRP solution can be fairly complicated even on simple metrics: for example, the optimum may cross itself several times on a line-metric~(\cite{AfratiCPPP86}) and the problem is NP-hard even on tree-metrics~(\cite{sitters2002minimum}).

Algorithm~\ref{makeconsecutive} describes the procedure used to establish Theorem~\ref{thm:consec}. We use $\Pi$ to denote the  distribution of active vertices, where each vertex has independent probability $p$. 

It is obvious that  each iteration of the while-loop decreases the number $k$ of parts of $S_z$: so this procedure ends in polynomial time and produces  a master tour that visits each $S_v$ consecutively. The key part of the proof is in showing that the expected latency does not increase.

\begin{algorithm}[H]
\caption{Algorithm to obtain a consecutive master tour.}\label{makeconsecutive}
\begin{algorithmic}[1]
\Statex \textbf{Procedure}\textsc{MakeConsecutive}({$\tau$}):
\For{$z\in V$}
    \State Let $C_z^1, C_z^2, ..., C_z^k$ be the minimal partition of $S_z$, where for every $i\in [k]$, the  vertices in $C_z^i$ appear consecutively  in tour $\tau$.
    \While{there exists $C_z^i$ and $C_z^j$ with $i \neq j$}
    \State Construct tour $\tau_i$ from $\tau$ by relocating vertices $C_z^j$ immediately after $C_z^i$
    \State Construct tour $\tau_j$ from $\tau$ by relocating vertices $C_z^i$ immediately before $C_z^j$
    \State $\tau\leftarrow \text{argmin}_{\tau\in\{\tau_i,\tau_j\}}\E_{A\gets \Pi}\left[ \lat_{\tau}^A\right]$
    \State Update $k\gets k-1$ and the  new partition  of $S_z$.
\EndWhile
\EndFor
\end{algorithmic}
\end{algorithm}

\begin{figure}[htp]
\centering
\includegraphics[scale = 0.5]{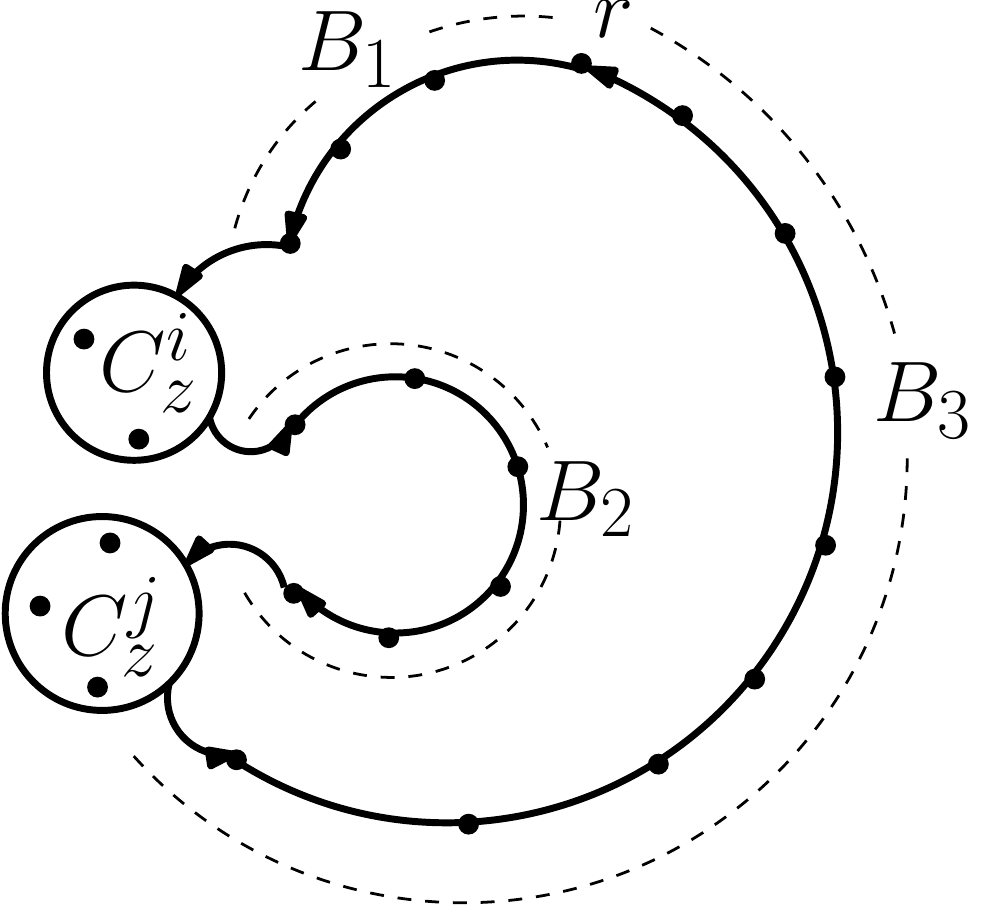}\hfill
\includegraphics[scale = 0.5]{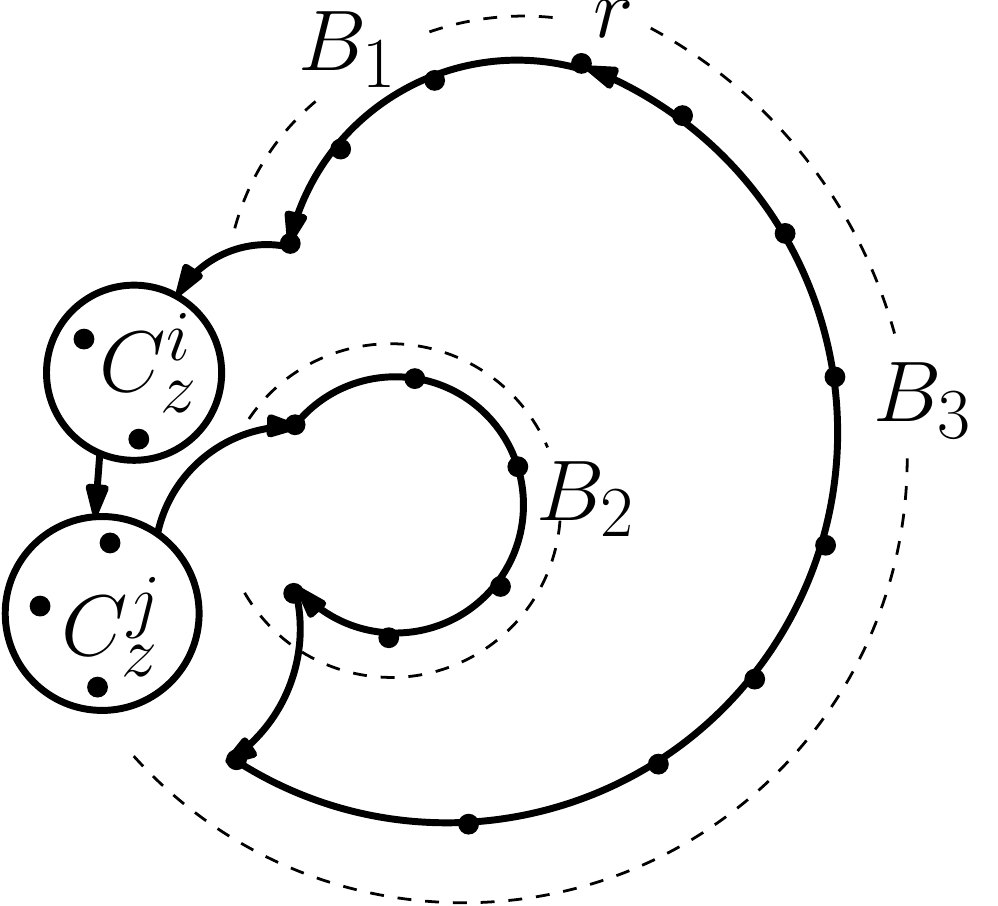}\hfill
\includegraphics[scale = 0.5]{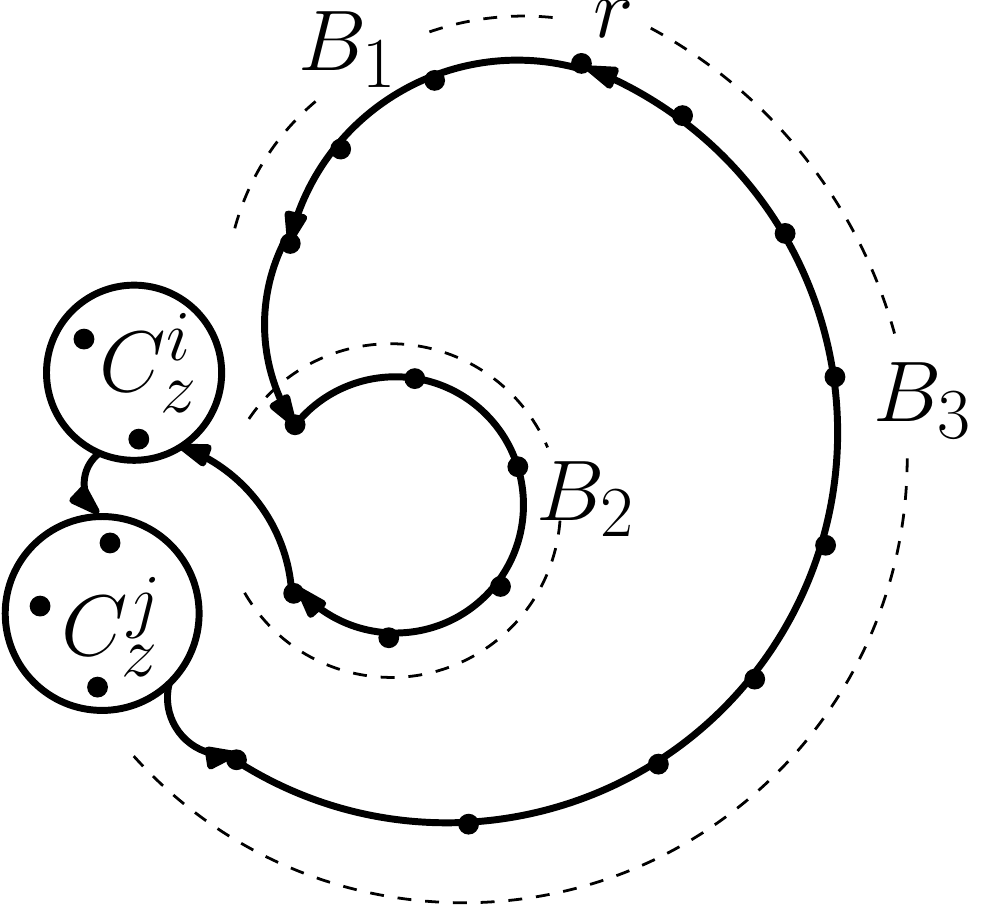}
\caption{From left to right: Tours $\tau, \tau_i$ and $\tau_j$}\label{tours}
\end{figure}

\begin{lemma} Let $C_{z}^i$ and $C_{z}^j$ be two parts of  $S_{z}$ with respect to the current tour $\tau$ in procedure \textsc{MakeConsecutive}. Then we have:
$$\E_{A\gets \Pi}\left[ \lat_{\tau}^A\right]\geq \min(\E_{A\gets \Pi}\left[ \lat_{\tau_i}^A\right] , \E_{A\gets \Pi}\left[ \lat_{\tau_j}^A\right]) \;.$$\label{lem:visitingconsec} 
\end{lemma}
\begin{proof}
Let  $|C_{z}^i| = k_i$ and $|C_{z}^j|=k_j$. 
Without loss of generality we assume that $\tau$ visits $C_{z}^i$ before $C_{z}^j$. To reduce notation we use $V$ to denote the vertex set of instance ${\cal J}$ and let $U=C_{z}^i \cup C_{z}^j$.  Recall that $\lat_{\pi}^A(w)$ is the latency of vertex $w$ in tour $\pi$ when the subset $A$ of vertices is active; also $\lat_{\pi}^A =\sum_{w\in A} \lat_{\pi}^A(w)$.
For any $R\sse V$ and $S\sse R$ we use the notation $p(S,R) = p^{|S|} \cdot (1-p)^{|R\setminus S|}$ for  the probability that $S$ is the set of active vertices amongst $R$.

It suffices to show $\E_{A\gets \Pi}\left[ \lat_{\tau}^A\right]$ is at least a convex combination of $\E_{A\gets \Pi}\left[ \lat_{\tau_i}^A\right]$ and $\E_{A\gets \Pi}\left[ \lat_{\tau_j}^A\right]$. More specifically we show that:
$$\E_{A\gets \Pi}\left[ \lat_{\tau}^A\right]\geq  \lambda \cdot \E_{A\gets \Pi}\left[ \lat_{\tau_i}^A\right] \,+\, (1-\lambda)\cdot   \E_{A\gets \Pi}\left[ \lat_{\tau_j}^A\right]\;.$$
where $\lambda\in [0,1]$ is a value that will be set later. 
The above inequality is equivalent to proving the following:
$$\sum_{\substack{A\subseteq V}}p(A,V)\lat_{\tau}^A\geq \sum_{\substack{A\subseteq V}}p(A,V)\left(\lambda \cdot \lat_{\tau_i}^A + (1-\lambda) \cdot \lat_{\tau_j}^A\right)\;.$$
Let us define $B=A\setminus U$ and $C = A\cap U$. Basically $C$ is the subset of active vertices among $U = C_z^i\cup C_z^j$, and $B$ is the subset of active vertices among the rest of $V$. Then we can re-write the above inequality as follows:
$$\hspace{-2mm}\small{\sum_{B\subseteq V\setminus U} p(B,V\setminus U) \sum_{\substack{C\subseteq U }} p(C,U)\lat_{\tau}^{B\cup C}\geq  \sum_{B\subseteq V\setminus U} p(B,V\setminus U) \sum_{\substack{C\subseteq U }} p(C,U)\left(\lambda \cdot \lat_{\tau_i}^{B\cup C}+ (1-\lambda) \cdot \lat_{\tau_j}^{B\cup C}\right)}\;.$$
 Therefore, it is enough to prove
\begin{equation}\label{weightedaverage}
\sum_{\substack{C\subseteq U}} p(C,U)\lat_{\tau}^{B\cup C}\geq \sum_{\substack{C\subseteq U}}  p(C,U)\left(\lambda \cdot  \lat_{\tau_i}^{B\cup C}+ (1-\lambda) \cdot  \lat_{\tau_j}^{B\cup C}\right),\quad \forall B\sse V\setminus U.
\end{equation}

\def\e{{\cal E}}

In the rest of this proof we fix a subset $B\sse V\setminus U$.  This can be viewed as conditioning on the event ``$B$ is the active set of vertices within $V\setminus U$''; we denote this event by $\e_B$. Let the order of visited vertices of $B\cup U$ in $\tau$ be $B_1, C_z^i, B_2, C_z^j, B_3$ where $B_1, B_2, B_3$ are ordered sets of vertices that form a partition of $B$. Therefore, together with $C_z^i$ and $C_z^j$ they form a partition of $B\cup U$. See Figure~\ref{tours} for an example. 

If $B_2=\emptyset$ then all three tours $\tau$, $\tau_i$ and $\tau_j$ become identical when restricted to $B\cup C$ for any $C\sse U$. So \eqref{weightedaverage} is satisfied with an equality in this case. Below we assume $B_2 \neq \emptyset$. 
We will prove the inequality~\eqref{weightedaverage} by considering the latency contributions of vertices in each of the 5 different parts $B_1, B_2, B_3, C_z^i, C_z^j$.

We define $l_w:=\lat_{\tau}^{B\cup \{w\}}(w)$ for all $w\in V$ and 
\begin{equation}\label{eq:consec-Tij}
 T_i:=\lat_{\tau}^{B\cup C_z^i}(w)\quad \forall w\in C_z^i, \quad \mbox{and} \quad T_j:=\lat_{\tau}^{B\cup C_z^j}(w)\quad \forall w\in C_z^j.
\end{equation}

Basically $T_i$ (resp. $T_j$) is the length of the path in $\tau$ from the root to any vertex in $C_z^i$ (resp. $C_z^j$) when the active vertices are $B\cup C_z^i$ (resp. $B\cup C_z^j$). Note that $T_j\ge T_i$ by triangle inequality. Also, let $L_{\pi}^B(w)$ be the expected latency of any vertex $w$ for any tour $\pi \in \{\tau, \tau_i, \tau_j\}$ conditioned on the event $\e_B$. More formally:
$$L_{\pi}^B(w) = \sum_{C\subseteq U}p(C,U)\lat_{\pi}^{B\cup C}(w),\qquad \forall w\in V\;.$$
Finally, defining the following terms will help us simplify  our notation:
\begin{equation}\label{eq:consec-Di}
\Delta_i := \lat_{\tau}^{B\cup C_{z}^i}(w) - \lat_{\tau}^B(w) = \lat_{\tau}^{B\cup C_{z}^i}(w) - l_w\hspace{1cm}\forall w\in B_2\cup B_3\;.\end{equation}

\begin{equation}\label{eq:consec-Dj}
\Delta_j := \lat_{\tau}^{B\cup C_{z}^j}(w) - \lat_{\tau}^B(w) = \lat_{\tau}^{B\cup C_{z}^j}(w) -l_w \hspace{1cm}\forall w\in B_3.
\end{equation}
Note that $\Delta_i$ (resp. $\Delta_j$) corresponds to  the increase in latency (conditioned on $\e_B$) of any vertex appearing after $C_{z}^i$ (resp. $C_{z}^j$) if some vertex in $C_{z}^i$ (resp. $C_{z}^j$) is active. 
Note that the right hand side in \eqref{eq:consec-Di} is the same for any $w$ in the given set and as a result independent of $w$; the same observation is true for \eqref{eq:consec-Dj}. Moreover, by triangle inequality, having a superset of active vertices can only increase the latency of any vertex: so $\Delta_i$ and $\Delta_j$ are non-negative.

Table~\ref{latencies} lists the expected latency of vertices in each of the five   different parts, conditioned on $\e_B$. 
We use  $\alpha_i = 1-{(1-p)}^{k_i}$ and $\alpha_j = 1-{(1-p)}^{k_j}$ as the probabilities of having at least one active vertex in parts $C_z^i$ and $C_z^j$ respectively. 

We first prove the lemma assuming the entries stated in the table. Then  we explain why each of these table entries  is correct, which would complete the proof. 

\begin{center}\begin{table}[h]
\begin{tabular}{ |c|c|c|c|c|c|c| }
 \hline
\backslashbox{Tour $\pi$}{Type} & $B_1$ & $B_2$ & $B_3$ & $C_z^i$ & $C_z^j$ \\
 \hline
$\tau$  & $l_w$  & $l_w + \Delta_i\alpha_i$
& $l_w + \Delta_i \alpha_i+\Delta_j \alpha_j$&$T_ip$&$T_jp + \Delta_i\alpha_ip$\\ \hline
$\tau_i$ & $l_w$ & $l_w + \Delta_i(\alpha_i+\alpha_j - \alpha_i\alpha_j)$& $l_w + \Delta_i(\alpha_i+\alpha_j - \alpha_i\alpha_j)$&$T_ip$&$T_ip$\\  \hline
$\tau_j$ & $l_w$& $l_w$& $l_w + \Delta_j(\alpha_i+\alpha_j - \alpha_i\alpha_j)$&$T_jp$&$T_jp$\\ \hline
\end{tabular}
\caption{The values of $L_{\pi}^B(w)$ for $w \in B_1\cup B_2\cup B_3\cup C_z^i\cup C_z^j$, and $\pi\in\{\tau, \tau_i, \tau_j\}$}\label{latencies}
\end{table}
\end{center}

\subsubsection{Completing  proof of Lemma~\ref{lem:visitingconsec} using Table~\ref{latencies}.} We now prove \eqref{weightedaverage} for a suitable choice of $\lambda\in [0,1]$. The value  $\lambda$ will not depend on the subset $B$: so (as discussed before) we can take an expectation over $B$ to complete the proof of the lemma.  

Choosing any $\lambda$ such that $\lambda\le \frac{\alpha_i}{\alpha_i+\alpha_j - \alpha_i \alpha_j}$ and $1-\lambda\le  \frac{\alpha_j}{\alpha_i+\alpha_j - \alpha_i \alpha_j}$, it follows from the first three columns of Table~\ref{latencies} (for $B_1$, $B_2$ and $B_3$) that:
\begin{equation}
\label{eq:consec-table1}
L^B_\tau(w) \ge \lambda \cdot L^B_{\tau_i}(w) + (1-\lambda) \cdot   L^B_{\tau_j}(w) , \qquad \forall w\in B.
\end{equation}

Next we show that the {\em total} latency contribution from $U$ satisfies a similar inequality: 
\begin{equation}
\label{eq:consec-table2}
\sum_{w\in U} L^B_\tau(w)    \ge \lambda \cdot \sum_{w\in U} L^B_{\tau_i}(w)+ 
(1-\lambda) \cdot \sum_{w\in U} L^B_{\tau_j}(w).
\end{equation}
To see this, note from the last two columns of the table that 
$$\sum_{w\in U} L^B_\tau(w) \ge  k_i\cdot T_i p + k_j\cdot T_j p, \quad \sum_{w\in U} L^B_{\tau_i}(w) = (k_i+k_j) T_i p, \quad \sum_{w\in U} L^B_{\tau_j}(w) = (k_i+k_j) T_j p\;.$$
So, to prove \eqref{eq:consec-table2} it suffices to show $k_i  T_i p + k_j  T_j p \, \ge\, (k_i+k_j) (\lambda T_i + (1-\lambda)T_j) p$. Using the fact that $T_i\le T_j$, it suffices to show $k_j \ge (k_i+k_j)(1-\lambda)$. In other words, choosing $\lambda$ such that $1-\lambda\le \frac{k_j}{k_i+k_j}$, we would obtain \eqref{eq:consec-table2}. 

Finally, adding the inequalities \eqref{eq:consec-table1} and \eqref{eq:consec-table2} (which account for the latency contribution from all active vertices) we would obtain \eqref{weightedaverage}. We only need to ensure that there is some choice for $\lambda$ satisfying the conditions we assumed, namely:
$$\lambda\le \frac{\alpha_i}{\alpha_i+\alpha_j - \alpha_i \alpha_j}, \quad  1-\lambda\le  \frac{\alpha_j}{\alpha_i+\alpha_j - \alpha_i \alpha_j}, \quad \mbox{and}\quad  1-\lambda\le \frac{k_j}{k_i+k_j}\;.$$
It can be verified directly that  $\lambda=\frac{1-(1-p)^{k_i}}{1-(1-p)^{k_i+k_j}}$ satisfies these conditions (see Appendix~\ref{app:lambda}).

\subsubsection{Obtaining the entries in Table~\ref{latencies}.} Below we consider each vertex-type separately. 
\paragraph{Vertices $w\in B_1$.}
By construction of $\tau_i$ and $\tau_j$ it is obvious that $\tau, \tau_i$ and $\tau_j$ are identical until visiting any $w\in B_1$. So for any $C\sse U$ and $\pi \in \{\tau, \tau_i, \tau_j\}$ we have $\lat_{\pi}^{B\cup C}(w) = \lat_{\tau}^B(w)=\lat_{\tau}^{B\cup \{w\}}(w)=l_w$. 
This means that $L_{\pi}^B(w) = l_w$ for all $\pi \in \{\tau, \tau_i, \tau_j\}$.

\paragraph{Vertices $w\in B_2$.}
Consider first tour $\tau$. Note that if there is at least one active vertex in $C_z^i$ (which happens with probability $\alpha_i$) then the latency of any $w\in B_2$ will be $\lat_{\tau}^{B\cup C_{z}^i}(w)$. However, if all vertices in $C_z^i$ are inactive (which happens with probability $1-\alpha_i$) then the latency of $w$   would be  $\lat_{\tau}^{B}(w)$. Now using \eqref{eq:consec-Di}  we have:
$$L_{\tau}^{B}(w) = \lat_{\tau}^{B\cup C_{z}^i}(w) \cdot \alpha_i + l_w\cdot (1-\alpha_i) = (l_w+\Delta_i)\cdot  \alpha_i  + l_w\cdot (1-\alpha_i) =l_w + \Delta_i \alpha_i\;.$$

Now, we can use a similar logic for $\tau_i$. Here, if there is any active vertex in $U=C_z^i \cup C_z^j$ (with probability $\alpha_i + \alpha_j - \alpha_i \alpha_j$) the  latency of $w$ is $\lat_{\tau_i}^{B\cup U}(w)$, and if all of $U$ is inactive the latency is $l_w$. Note that by definition of $\tau$ and $\tau_i$ and the fact that all vertices in $C_z^i$ appear consecutively on both tours, $\lat_{\tau_i}^{B\cup U}(w) = \lat_{\tau_i}^{B\cup C_{z}^i}(w) = \lat_{\tau}^{B\cup C_{z}^i}(w)$. So we have   
$L_{\tau_i}^{B} = l_w + \Delta_i (\alpha_i + \alpha_j -\alpha_i\alpha_j)$.

Finally, by definition of $\tau_j$  we have $\lat_{\tau_j}^{B\cup C}(w) = \lat_{\tau}^{B}(w) = l_w$ for any $C\sse U$. So  $L_{\tau_j}^B(w) = l_w$.

\paragraph{Vertices $w\in B_3$.} Consider first tour $\tau$. The latency of such a vertex $w$ is:
\begin{itemize}
\item $l_w$ if all of $U=C_z^i\cup C_z^j$ is inactive,
\item $\lat_{\tau}^{B\cup C_z^i}(w)$ if some vertex in $C_z^i$ is active and all of $C_z^j$ is inactive,
\item $\lat_{\tau}^{B\cup C_z^j}(w)$ if some vertex in $C_z^j$ is active and all of $C_z^i$ is inactive, and
\item $\lat_{\tau}^{B\cup C_z^i\cup C_z^j}(w)$ if some vertex in $C_z^i$ and some vertex in $C_z^j$ are active.
\end{itemize}
Therefore, we can write:
$$L_{\tau}^{B}(w) = l_w(1-\alpha_i)(1-\alpha_j) +  \lat_{\tau}^{B\cup C_z^i}(w)\alpha_i(1-\alpha_j)+  \lat_{\tau}^{B\cup C_z^j}(w)\alpha_j(1-\alpha_i) + \lat_{\tau}^{B\cup U}(w)\alpha_i\alpha_j\;.$$
From~\eqref{eq:consec-Di} and \eqref{eq:consec-Dj} we have   $\lat_{\tau}^{B\cup C_z^i} = l_w + \Delta_i$ and $\lat_{\tau}^{B\cup C_z^j}=l_w + \Delta_j$. 
Also, since  we assumed that $B_2\neq \emptyset$, we have $\lat_{\tau}^{B\cup U} = l_w + \Delta_j  + \Delta_i$. Combined with the above equation, 
$$L_{\tau}^{B}(w) = l_w +\Delta_i\alpha_i +\Delta_j\alpha_j\;.$$

Now for tour $\tau_i$ the latency would be equal to $\lat_{\tau}^{B\cup C_z^i}(w) = l_w + \Delta_i$ if there is at least one active vertex among $U$ which happens with probability $\alpha_i + \alpha_j -\alpha_i\alpha_j$. Otherwise it would be just $l_w$.
So  $L_{\tau_i}^{B} =  l_w + \Delta_i(\alpha_i + \alpha_j -\alpha_i\alpha_j)$. 
Similarly, for tour $\tau_j$ we have  $L_{\tau_j}^{B} = l_w + \Delta_j(\alpha_i + \alpha_j -\alpha_i\alpha_j) $.

\paragraph{Vertices $w \in C_z^i$.}
We  start with tour $\tau$. If $w\notin C$ then $\lat_{\tau}^{B\cup C}(w) = 0$.
Otherwise, $w$ is active and using \eqref{eq:consec-Tij}  we have $\lat_{\tau}^{B\cup C}(w) = \lat_{\tau}^{B\cup C_z^i}(w) = T_i$.  So  $L_{\tau}^B(w) = T_i p$.

As $C_z^i$ appears in the same position in tours $\tau$ and $\tau_i$,   we also have   $L_{\tau_i}^{B}(w) = T_i p$. 

In tour $\tau_j$, part $C_z^i$ has moved to the position of part $C^j_z$ in $\tau$. Here, when $w\in C$ we have $\lat_{\tau_j}^{B\cup C}(w) =  T_j$. So  $L_{\tau_j}^{B}(w) = T_j p$.

\paragraph{Vertices $w \in C_z^j$.} As in the  previous case, we have $L_{\tau_i}^B(w) = T_ip$ and $L_{\tau_j}^B(w) = T_jp$. 

Now, consider tour $\tau$. First note that if $w\notin C$, $\lat_{\tau}^{B\cup C}(w) = 0$. Below we consider the cases that $w$ is active, which happens with probability $p$. If there is at least one active vertex in $C_z^i$ (which happens independently with probability $\alpha_i$) we have $\lat_{\tau}^{B\cup C}(w) = l_w + \Delta_i = T_j + \Delta_i$. And if there is no active vertex in $C_z^i$ (with probability $1-\alpha_i$), then we have $\lat_{\tau}^{B\cup C}(w) = l_w = T_j$. So 
$$L_{\tau}^B(w) = p\alpha_i \cdot (T_j + \Delta_i) + p(1- \alpha_i) \cdot T_j = T_jp + \Delta_i \alpha_ip\;.$$

This completes the proof of all cases in Table~\ref{latencies}, and hence Lemma~\ref{lem:visitingconsec}.

\end{proof}

\paragraph{Acknowledgement}
V. Nagarajan and F. Navidi were supported in part by NSF CAREER grant CCF-1750127. I.~L.~G{\o}rtz was supported in part by the Danish Research Council grant DFF -- 1323-00178.

\bibliographystyle{plain}

\begin{appendices}

\section{Choice of $\lambda$ in proof of Lemma~\ref{lem:visitingconsec}} \label{app:lambda}
Here we show that  $\lambda=\frac{1-(1-p)^{k_i}}{1-(1-p)^{k_i+k_j}}$ (where $0\le p\le 1$) satisfies the following inequalities:
\begin{gather}
\lambda\le \frac{\alpha_i}{\alpha_i+\alpha_j - \alpha_i \alpha_j}\label{lambda:alpha:i}\\
1-\lambda\le  \frac{\alpha_j}{\alpha_i+\alpha_j - \alpha_i \alpha_j} \label{lambda:alpha:j}\\
1-\lambda\le \frac{k_j}{k_i+k_j}\label{lambda:k:j}
\end{gather}
where $\alpha_i = 1-{(1-p)}^{k_i}$ and $\alpha_j = 1-{(1-p)}^{k_j}$.

We define function $f(k) = 1-(1-p)^k$. Then we can write:
$$\lambda = \frac{f(k_i)}{f(k_i + k_j)},\quad\alpha_i = f(k_i), \quad \alpha_j = f(k_j)$$
Clearly, 
\begin{equation}
f(k_i+k_j) = f(k_i) + f(k_j) - f(k_i)f(k_j)\label{inc:exc}
\end{equation}
Now, we can re-write inequality~\eqref{lambda:alpha:i} as:
$$\frac{f(k_i)}{f(k_i + k_j)} \leq \frac{f(k_i)}{f(k_i) + f(k_j) - f(k_i)f(k_j)}$$
which is true by equation~\eqref{inc:exc}.

For inequality~\eqref{lambda:alpha:j}, we rewrite it as:
$$1-\frac{f(k_i)}{f(k_i + k_j)} \leq \frac{f(k_j)}{f(k_i) + f(k_j) - f(k_i)f(k_j)} = \frac{f(k_j)}{f(k_i+k_j)}$$
$$\Leftrightarrow\quad  f(k_i+k_j)\leq f(k_i) + f(k_j)\;,$$
which is true by~\eqref{inc:exc} and the fact that $f(k)\geq 0$ for every $k$.

It  remains to show the correctness of inequality~\eqref{lambda:k:j} which can be written as:
$$\frac{f(k_i)}{f(k_i + k_j)}\geq\frac{k_i}{k_i + k_j} \quad \Leftrightarrow\quad   \frac{f(k_i)}{k_i}\geq\frac{f(k_i + k_j)}{k_i + k_j}\;.$$
So it is enough to show that $g(k) = \frac{f(k)}{k}$ is decreasing, or equivalently $g'(k) \leq 0$. We can write:
$$g'(k) = \frac{kf'(k)-f(k)}{k^2} = \frac{(1-p)^k(1-k\log(1-p))-1 }{k^2} \le \frac{(1-p)^k\cdot e^{-k\log(1-p)}-1 }{k^2} = 0\;.$$
Above we used the inequality $1+x\le e^x$ for all real $x$. 

\end{appendices}

\end{document}